\documentclass[11pt,a4paper]{article}

\usepackage[utf8]{inputenc}
\usepackage[english]{babel}

\usepackage{geometry}
\usepackage{graphicx}
\usepackage{amssymb,amsmath,amsthm}
\usepackage{delarray}
\usepackage{enumitem}
\usepackage{tikz}
\usepackage{hyperref}
\usepackage{authblk}

\newtheorem{theorem}{\textbf{Theorem}}

\newtheorem{corollary}{\textbf{Corollary}}

\newtheorem{remark*}{Remark}

\newcommand{\cycle}[1]{\mathcal C(#1)}
\newcommand{\VC}[1]{\mathcal{VC}(#1)}
\newcommand{\FAS}[1]{\mathcal{F\!\!AS}(#1)}
\renewcommand{\O}{\mathcal{O}}
\newcommand{\Poly}{{\mathsf{P}}}
\newcommand{\sPoly}{{\mathsf{\#P}}}
\newcommand{\sOptPoly}{{\mathsf{\#\!\cdot\!OptP[\log n]}}}
\newcommand{\NP}{{\mathsf{NP}}}
\newcommand{\redparsi}{\leq^p_\text{parsi}}
\newcommand{\redT}{\leq^p_T}
\newcommand{\countingpb}[3]{\fbox{\parbox{0.9\textwidth}{{\bf #1}\\{\it Input:} #2\\{\it Output:} #3}}}

\title{\vspace*{-1cm}On the complexity of counting feedback arc sets}
\author[$\dagger$]{K\'evin Perrot}
\affil[$\dagger$]{Universit\'e Publique, France.}
\date{}

\begin{document}
\setlist[itemize,enumerate]{nosep}
\maketitle

%%%%%%%%%%%%%%%%%%%%%%%%%%
\begin{abstract}
  In this note we study the computational complexity of feedback arc set
  counting problems in directed graphs, highlighting some subtle yet common
  properties of
  counting classes. Counting the number of feedback arc sets of cardinality $k$
  and the total number of feedback arc sets are $\sPoly$-complete problems,
  while counting the number of minimum feedback arc sets is only proven to be
  $\sPoly$-hard. Indeed, this latter problem is $\sOptPoly$-complete, hence
  if it belongs to $\sPoly$ then $\Poly=\NP$.
  %\\[.5em]
  %N.B.: Despite some bibliographic investigations I haven't found
  %these results stated anywhere,
  %though they are derived from
  %fairly standard methods and may be useful. 
  %If they are already published
  %please let me know: \texttt{name.surname@lis-lab.fr}!
\end{abstract}

%%%%%%%%%%%%%%%%%%%%%%%%%%
\section{Introduction}
\label{s:intro}

Feedback arc sets are natural objects to consider in digraphs, and deciding,
given a digraph $G$ and an integer $k$, whether it contains a feedback arc set
of cardinality at most $k$ is one of Karp's 21 $\NP$-complete problems
\cite{k72} (with a reduction not parsimonious).
In this note we study the computational complexity of various 
feedback arc set counting problems. Section~\ref{s:prelim} gives the required
definitions, in Section~\ref{s:complete} we prove that counting the number of
feedback arc sets of cardinality $k$ and counting the total number of feedback
arc sets are $\sPoly$-complete problems, in Section~\ref{s:minimum} we study
the special case of counting the number of minimum feedback arc set (which is
not in $\sPoly$ unless $\Poly=\NP$), and Section~\ref{s:fvs} remarks that we
can derive identical results on counting feedback vertex sets.

%%%%%%%%%%%%%%%%%%%%%%%%%%
\section{Preliminaries}
\label{s:prelim}

We denote $[n]$ the set of integers $\{1,\dots,n\}$. We call {\em graph} an
undirected graph $G=(V,E)$ with $E \subseteq \{ \{u,v\} \mid u,v \in V \}$, and
{\em digraph} a directed graph $G=(V,A)$ with $A \subseteq V \times V$.
We consider all our graphs and digraphs to be loopless. A {\em
vertex cover} $C \subseteq V$ of a graph $G=(V,E)$ is a subset of vertices
intersecting every edge of $G$, {\em i.e.} verifying $\forall e \in E: e \cap
C \neq \emptyset$. A
{\em cycle of length $k$} in a digraph $G=(V,A)$ is a $k$-tuple of arcs
$((u_1,v_1),(u_2,v_2),\dots,(u_k,v_k))$ such that $(u_i,v_i) \in A$ for all $i
\in [k]$, $v_i = u_{i+1}$ for all $i \in [k-1]$, and $u_1=v_k$. Let $\cycle{G}$
denote the set of cycles of $G$. If $\cycle{G}=\emptyset$ then $G$ is called
{\em acyclic}. A {\em feedback arc set} $F \subseteq A$ of a digraph $G=(V,A)$
is a subset of arcs intersecting every cycle of $G$, {\em i.e.} verifying
$\forall c \in \cycle{G}: c \cap F \neq \emptyset$. Let $\FAS{G}$
(resp. $\VC{G}$) denote the set of feedback arc sets (FAS)
(resp. vertex covers (VC)) of $G$.

Valiant defined the {\em class of counting problems
$\sPoly$} in \cite{v79,v79b}.
%We follow the formalism of \cite{p14} and define
%$\sPoly$ as the class of functions $f:\Sigma^* \to \N$ for some finite alphabet
%$\Sigma$, such that there exists a
%langage $A \in \Poly$ and a polynomial $p(n)$ with:
%$$\forall x \in \Sigma^*: f(x)=|\{y\in \Sigma^{p(|x|)}|(x,y)\in A \}|,$$
%with $|X|$ the number of elements (cardinality) of set $X$, and $|x|$ the
%number of letters of a word $x$.
It is the class of functions counting the number of
certificates of decision problems in $\NP$, {\em i.e.} problems
of the form ``given $x$, compute $f(x)$'', where $f$ is
the number of accepting paths
of a nondeterministic Turing machine taking $x$ as input and
running in polynomial time. We denote
$f \redparsi g$ when there exists a {\em polynomial parsimonious reduction} from
counting problem $f$ on alphabet $\Sigma_f$
to counting problem $g$ on alphabet $\Sigma_g$, {\em
i.e.} a transformation $r:\Sigma_f^* \to \Sigma_g^*$
computable by a deterministic Turing machine running in polynomial time
such that $\forall x \in
\Sigma_f^*: f(x)=g(r(x))$. We denote $f \redT g$ when there exists a {\em polynomial
Turing reduction} from $f$ to $g$, {\em i.e.} a deterministic Turing machine 
computing $f$ with oracle $g$, and running in polynomial time (hence making a
polynomial number of calls to its oracle). Of course $f \redT g$
implies $f \redparsi g$.

$\sPoly$-hardness is defined using either polynomial parsimonious reductions or
polynomial Turing reductions (both reductions yield the same results).
Note that, in contrast,
Turing reductions are too strong for counting classes that are higher
in the polynomial counting hierarchy \cite{tw92}. Furthermore, 
$\sPoly$ is known to be closed under $\redparsi$, but this question is open for
$\redT$ and it is known to be equivalent to $\Poly=\NP$ (see
Section~\ref{s:minimum}).

The authors of \cite{hp09} have developed a remedy to the difficulty of classifying
tightly some problems in $\sPoly$ (related to the unknown closure of $\sPoly$
for $\redT$), in particular minimality and maximality counting problems, that we
will apply in Section~\ref{s:minimum}. It consists in the {\em counting class
$\sOptPoly$}, which is defined in terms of nondeterministic Turing machines 
where each accepting path outputs a number encoded in binary
(rejecting paths have no output).
We will call such machines {\em nondeterministic transducers}.
Then $\sOptPoly$ is the class of functions counting the number of accepting
paths outputting the minimal value (among all the values outputted by accepting paths), for some
nondeterministic transducer running in polynomial time and outputting
values whose binary encoding is of length $\O(\log n)$
(the definitions of \cite{hp09} are more general
in order to fit other levels of
the polynomial counting hierarchy
and other magnitudes of output length).
In this paper we will only use the notion of hardness for $\sOptPoly$ under
polynomial parsimonious reductions.

Our complexity results will be derived by reduction from
the following problems.\\[.5em]
\countingpb{Cardinality vertex cover (\#Card-VC)}
{A graph $G$ and an integer $k$.}
{$|\{ C \in \VC{G} \mid |C|=k \}|$.}

\begin{theorem}[{\cite[p.~169]{gj79}}]
  \label{th:card-vc}
  {\bf \#Card-VC} is $\sPoly$-complete.
\end{theorem}

\noindent
  \countingpb{Minimum cardinality vertex cover (\#Minimum-VC)}
{A graph $G$.}
{$|\{ C \in \VC{G} \mid |C|=m \}|$ with $m = \min\{ |C| \mid C \in \VC{G} \}$.}

\begin{theorem}[{\cite[by identity reduction from Problem~4]{pb83}}]%[{\cite[Theorem~4.2]{hmrs98}}]
  \label{th:minimum-vc-spoly}
  {\bf \#Minimum-VC} is $\sPoly$-hard.
\end{theorem}

\begin{theorem}[{\cite[Theorem~11]{hp09}}]
  \label{th:minimum-vc-soptpoly}
  {\bf \#Minimum-VC} is $\sOptPoly$-complete.
\end{theorem}

%%%%%%%%%%%%%%%%%%%%%%%%%%
\section{$\sPoly$-complete feedback arc set problems}
\label{s:complete}

\countingpb{Cardinality feedback arc set (\#Card-FAS)}
{A digraph $G$ and an integer $k$.}
{$|\{ F \in \FAS{G} \mid |F|=k \}|$.}\\[.5em]
\countingpb{Feedback arc set (\#FAS)}
{A digraph $G$.}
{$|\FAS{G}|$.}\\[.5em]

The first reduction adapts the classical construction from Karp \cite{k72}.

\begin{theorem}
  \label{th:card-fas}
  {\bf \#Card-VC} $\redT$ {\bf \#Card-FAS},
  and {\bf \#Card-FAS} is $\sPoly$-complete.
\end{theorem}

\begin{proof}
  {\bf \#Card-FAS} is in $\sPoly$ since given $G=(V,A)$ and $k$,
  one can guess nondeterministically a subset $F \subseteq A$ of size $k$,
  and then check in polynomial time that the digraph $(V,A\setminus F)$
  is acyclic.

  \medskip

  For the $\sPoly$-hardness, as claimed we construct a polynomial Turing
  reduction from {\bf \#Card-VC} which is $\sPoly$-hard (Theorem~\ref{th:card-vc}).
  Given an instance $G=(V,E)$ and $k$ of {\bf \#Card-VC}, we consider 
  an arbitrary order $\prec$ on $V$ and the
  digraph $G'(\ell)=(V'(\ell),A'(\ell))$ with
  $$
    V'(\ell)= \{ v_i \mid v \in V \text{ and } i \in \{0,1\} \} \cup
    \{ e_{i,j} \mid e \in E \text{ and } i \in \{0,1\} \text{ and } j \in [\ell] \},
  $$
  $$
    \begin{array}{r}
      %A'(\ell)= \{ (v_0,v_1) \mid v \in V \}\!\!\!\!\! 
      %& \cup~ \{ (u_1,\{u,v\}_{0,j}) \mid \{u,v\} \in E \text{ and } u \prec v \text{ and } j \in [\ell] \}\\
      %& \cup~ \{ (\{u,v\}_{0,j},v_0) \mid \{u,v\} \in E \text{ and } u \prec v \text{ and } j \in [\ell] \}\\
      %& \cup~ \{ (v_1,\{u,v\}_{1,j}) \mid \{u,v\} \in E \text{ and } u \prec v \text{ and } j \in [\ell] \}\\
      %& \cup~ \{ (\{u,v\}_{1,j},u_0) \mid \{u,v\} \in E \text{ and } u \prec v \text{ and } j \in [\ell] \},
      A'(\ell)=\{ (v_0,v_1) \mid v \in V \} \cup \Big\{ (u_1,\{u,v\}_{0,j}),(\{u,v\}_{0,j},v_0),(v_1,\{u,v\}_{1,j}),(\{u,v\}_{1,j},u_0) \mid\\ \{u,v\} \in E \text{ and } u \prec v \text{ and } j \in [\ell] \Big\},
    \end{array}
  $$
  for $\ell=k+1$ (the construction is illustrated on Figure~\ref{fig:card-fas}).
  \begin{figure}
    %\centerline{\includegraphics{fig-card-fas.pdf}}
    \centerline{\begin{tikzpicture}
  % variables
  \def\ns{.4} % vertex node shift
  \def\as{.3} % arc node shift
  % styles
  \tikzstyle{every node}=[draw,circle,inner sep=2pt]
  % original graph
  \begin{scope}[shift={(-6.5,0)}]
    \node (a) at (1,0) {$a$};
    \node (b) at (0,-1) {$b$};
    \node (c) at (-1,0) {$c$};
    \node (d) at (0,1) {$d$};
    \draw (a) to (b);
    \draw (b) to (c);
    \draw (c) to (d);
    \draw (d) to (a);
    \node[draw=none] at (-1.7,0) {$G=$};
  \end{scope}
  % vertices v
  \node[shift={(-\ns,0)}] (a0) at (0:2) {$a_0$};
  \node[shift={(\ns,0)}] (a1) at (0:2) {$a_1$};
  \node[shift={(0,-\ns)}] (b0) at (-90:2) {$b_0$};
  \node[shift={(0,\ns)}] (b1) at (-90:2) {$b_1$};
  \node[shift={(\ns,0)}] (c0) at (180:2) {$c_0$};
  \node[shift={(-\ns,0)}] (c1) at (180:2) {$c_1$};
  \node[shift={(0,\ns)}] (d0) at (90:2) {$d_0$};
  \node[shift={(0,-\ns)}] (d1) at (90:2) {$d_1$};
  % vertices e
  \foreach \x in {1,2,3}{
    \node (ab\x) at (-45:2+\as*\x) {};
    \node (ba\x) at (-45:2-\as*\x) {};
    \node (bc\x) at (180+45:2-\as*\x) {};
    \node (cb\x) at (180+45:2+\as*\x) {};
    \node (cd\x) at (90+45:2+\as*\x) {};
    \node (dc\x) at (90+45:2-\as*\x) {};
    \node (da\x) at (45:2-\as*\x) {};
    \node (ad\x) at (45:2+\as*\x) {};
  }
  % arcs
  \foreach \v in {a,b,c,d}
    \draw[->] (\v0) to (\v1);
  \foreach \v/\vv in {a/b,b/c,c/d,d/a}
    \foreach \x in {1,2,3}{
      \draw[->] (\v1) to (\v\vv\x);
      \draw[->] (\v\vv\x) to (\vv0);
      \draw[->] (\vv1) to (\vv\v\x);
      \draw[->] (\vv\v\x) to (\v0);
    }
  % labels
  \draw[rotate=-45,dashed] (2.6,0) ellipse (.6 and .4);
  \node[shift={(.5,0)},right,draw=none] at (-45:2.5) {\scriptsize $\{a,b\}_{0,i}$ for $i \in [3]$};
  \node[draw=none] at (-3.5,0) {$G'(3)=$};
  % order
  \node[draw=none] at (0,0) {\scriptsize $a \prec b \prec c \prec d$};
\end{tikzpicture}}
    \caption{Illustration of the construction $G'(\ell)$ for $\ell=3$ in the
    proof of Theorem~\ref{th:card-fas}.}
    \label{fig:card-fas}
  \end{figure}
  First, note that $k$ is encoded in binary, but since $k \leq |V|$ the 
  construction is polynomial in size. 
  The idea is to have, for each edge of $G$,
  a large set of cycles in $G'(\ell)$.
  In $G'(\ell)$ there are $4\ell|E|$ arcs
  corresponding to edges of $G$ that are ``strictly less interesting''
  (when $\ell > 1$)\footnote{If $\ell=1$ then $k=0$ and the answer to 
  the {\bf \#Card-VC} instance is trivial.}
  to construct an
  FAS than the $|V|$ arcs corresponding to vertices of $G$,
  because any cycle cut by some former edge can also be cut by some latter edge. 
  As a consequence, there is a one-to-one correspondence between the
  minimum VC of $G$ and the minimum FAS of $G'(\ell)$.
  Let $C(\kappa)$ be the answer to {\bf \#Card-VC} on $G$ for some integer $\kappa$,
  let $F(k')$ be the answer to {\bf \#Card-FAS} on $G'(\ell)$ for 
  some integer $k'$, and let
  $$m=\min\{ |C| \mid C \in \VC{G} \}=\{ |F| \mid F \in \FAS{G'(\ell)} \}.$$
  When $\ell > k' \geq m$, an FAS of size $k'$ of $G'(\ell)$ can be seen as:
  \begin{itemize}
    \item $\kappa$ arcs corresponding to vertices of $G$ forming a vertex cover
      for some $k' \geq \kappa \geq m$,
    \item $k'-\kappa$ arcs corresponding to edges of $G$ (among the $4\ell|E|$),
  \end{itemize}
  because with $\ell>k'$ it is not possible to intersect all cycles
  corresponding to an edge of $G$ with only arcs corresponding to edges of $G$.
  We can deduce the following relations between $C(\kappa)$ and $F(k')$:
  $$
    \begin{array}{l}
      F(0)=C(0)=0\\
      \dots\\
      F(m-1)=C(m-1)=0\\
      F(m)=C(m)\\
      F(m+1)=C(m+1)+{4\ell|E| \choose 1}C(m)\\
      F(m+2)=C(m+2)+{4\ell|E| \choose 1}C(m+1)+{4\ell|E| \choose 2}C(m)\\
      \dots\\
      F(m+i)=C(m+i)+\sum_{j=0}^{i-1} {4\ell|E| \choose i-j} C(m+j).
    \end{array}
  $$
  It is therefore possible to compute inductively $C(\kappa)$ from oracle calls to
  $F(1),\dots,F(\kappa)$ (corresponding to {\bf \#Card-FAS} instances
  $G'(\ell)$ with integers $k'=1,\dots,\kappa$),
  and eventually compute $C(k)$ after a polynomial time
  since all calls are done for $k'\leq k \leq |V|$.
\end{proof}

The second reduction employs the Vandermonde matrix method
from \cite{pb83,v79}.

\begin{theorem}
  \label{th:fas}
  {\bf \#Card-FAS} $\redT$ {\bf \#FAS},
  and {\bf \#FAS} is $\sPoly$-complete.
\end{theorem}

\begin{proof}
  Again {\bf \#FAS} is in $\sPoly$ since given $G=(V,A)$,
  one can guess nondeterministically a subset $F \subseteq A$,
  and then check in polynomial time that the digraph $(V,A\setminus F)$
  is acyclic.

  \medskip

  For the $\sPoly$-hardness, as claimed we construct a polynomial Turing
  reduction from {\bf \#Card-FAS} which is $\sPoly$-hard (Theorem~\ref{th:card-fas}).
  Given an instance $G=(V,A)$ and $k$ of {\bf \#Card-FAS}, we consider 
  the family of digraphs $H'(\ell)=(V'(\ell),A'(\ell))$ with
  $$
    V'(\ell)= V \cup \{ a_{i} \mid a \in A \text{ and } i \in [\ell-1] \},
  $$
  $$
    A'(\ell)= \{ (u,a_1),(a_1,a_2),\dots,(a_{\ell-2},a_{\ell-1}),(a_{\ell-1},v) \mid a=(u,v) \in A \} 
  $$
  (the construction is illustrated on Figure~\ref{fig:fas}).
  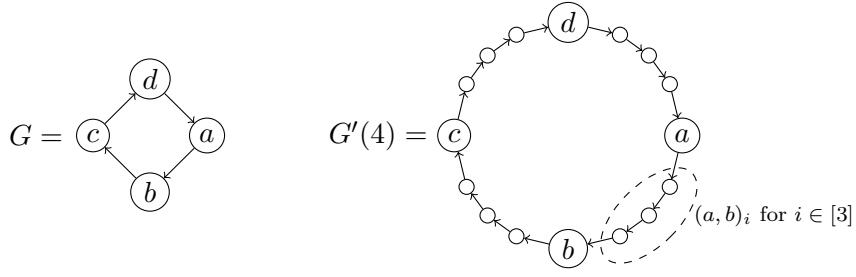
\begin{figure}
    %\centerline{\includegraphics{fig-fas.pdf}}
    \centerline{\begin{tikzpicture}
  % variables
  \def\vs{.75} % vertex sep to origin (original)
  \def\ns{1.5} % vertex node sep to origin
  \def\as{18} % arc node shift in degree
  % styles
  \tikzstyle{every node}=[draw,circle,inner sep=2pt]
  % original graph
  \begin{scope}[shift={(-5.5,0)}]
    \node (a) at (\vs,0) {$a$};
    \node (b) at (0,-\vs) {$b$};
    \node (c) at (-\vs,0) {$c$};
    \node (d) at (0,\vs) {$d$};
    \draw[->] (a) to (b);
    \draw[->] (b) to (c);
    \draw[->] (c) to (d);
    \draw[->] (d) to (a);
    \node[draw=none] at (-1.5,0) {$G=$};
  \end{scope}
  % vertices
  \node (a) at (\ns,0) {$a$};
  \node (b) at (0,-\ns) {$b$};
  \node (c) at (-\ns,0) {$c$};
  \node (d) at (0,\ns) {$d$};
  \foreach \v/\d in {a/-45,b/180+45,c/90+45,d/45}{
    \node (\v1) at (\d+\as:\ns) {};
    \node (\v2) at (\d:\ns) {};
    \node (\v3) at (\d-\as:\ns) {};
  }
  % arcs
  \foreach \v/\vv in {a/b,b/c,c/d,d/a}{
    \draw[->] (\v) to (\v1);
    \draw[->] (\v1) to (\v2);
    \draw[->] (\v2) to (\v3);
    \draw[->] (\v3) to (\vv);
  }
  % labels
  \draw[rotate=-45,dashed] (\ns,0) ellipse (.4 and .8);
  \node[shift={(.5,0)},right,draw=none] at (-45:\ns) {\scriptsize $(a,b)_i$ for $i \in [3]$};
  \node[draw=none] at (-2.5,0) {$G'(4)=$};
\end{tikzpicture}}
    \caption{Illustration of the construction $H'(\ell)$ for $\ell=4$ in the
    proof of Theorem~\ref{th:fas}.}
    \label{fig:fas}
  \end{figure}
  The idea is that an arc $a=(u,v)$ of $G$ is replaced with a path 
  of $\ell$ arcs in $H'(\ell)$ that we will denote
  $$
    P_a=\{(u,a_1),(a_1,a_2),\dots,(a_{\ell-2},a_{\ell-1}),(a_{\ell-1},v)\}.
  $$
  %hence for each FAS $F$ of $G$
  %such that $a \notin F$
  %will correspond to FAS $F'$ of $G$ such that
  %$P_a \cap F' = \emptyset$.
  %More precisely,
  Hence for each FAS $F$ of $G$ there corresponds
  a family $\Omega(F)$ of FAS of $H'(\ell)$, with elements of the form
  $\bigcup_{a \in A} S'_a$ where
  $$
    \begin{array}{l}
      S'_a \cap P_a \neq \emptyset \text{ if } a \in F\\
      S'_a \cap P_a = \emptyset \text{ if } a \notin F.
    \end{array}
  $$
  The family $\Omega(F)$ consists of $(2^\ell-1)^{|F|}$ FAS, and the families
  $\{\Omega(F) \mid F \in \FAS{G} \}$ partition $\FAS{H'(\ell)}$.
  The number of FAS of $H'(\ell)$ is therefore
  $$
    \Gamma(\ell)=\sum_{i=0}^{|A|} F_i(G) (2^\ell-1)^i
  $$
  where $F_i(G)$ is the number of FAS of $G$ of cardinality
  $i \in \{0,\dots,|A|\}$, {\em i.e.}
  $$
    F_i(G)=|\{ F \in \FAS{G} \mid |F|=i \}|.
  $$
  The $(|A|+1)\times(|A|+1)$ matrix $M=(M_{\ell i})$ with entries
  $M_{\ell i}=(2^\ell-1)^i \text{ for } \ell,i \in \{0,\dots,|A|\}$
  is Vandermonde\footnote{A matrix (or its transpose)
  is Vandermonde when its coefficient
  on the $i^\text{th}$ row and $j^\text{th}$ column
  is of the form $\mu_i^j$ for some real numbers $\mu_i$
  (indices start with $0$ hence coefficients equal 1 on the first column).}
  with bases distinct for all $\ell$,
  as a consequence it is nonsingular and from $\Gamma(\ell)$ for
  $\ell \in \{0,\dots,|A|\}$ we can compute $F_i(G)$ for
  $i \in \{0,\dots,|A|\}$ in polynomial time
  (see \cite[Lemma page~781]{pb83}).
  %deduce from \cite[Chapter~5.1]{hlp52}
\end{proof}

Let us also mention the following variant.\\[.5em]
\countingpb{Minimal feedback arc set (\#Minimal-FAS)}
{A digraph $G$.}
{$|\{ F \in \FAS{G} \mid \forall F' \subsetneq F : F' \notin \FAS{G} \}|$.}\\[.5em]
It is proven in \cite{ss02} that {\bf \#Minimal-FAS} is $\sPoly$-complete, with
a parsimonious reduction from counting the number of
acyclic orientations of an undirected
graph, itself proven to be $\sPoly$-hard in \cite{l86}. {\bf \#Minimal-FAS}
belongs to $\sPoly$ since the set of feedback arc sets is obviously
closed by arc addition, hence
one can check in polynomial time
that for any $a \in F$ the digraph $(V,A\setminus (F
\setminus \{a\}))$ is not acyclic.

\begin{theorem}[{\cite[Corollary~1]{ss02}}]
  \label{th:minimal-fas}
  {\bf \#Minimal-FAS} is $\sPoly$-complete.
\end{theorem}

%%%%%%%%%%%%%%%%%%%%%%%%%%
\section{Counting minimum feedback arc sets}
\label{s:minimum}

\countingpb{Minimum feedback arc set (\#Minimum-FAS)}
{A digraph $G$.}
{$|\{ F \in \FAS{G} \mid |F|=m \}|$
with $m=\min\{ |F| \mid F \in \FAS{G} \}$.}\\[.5em]

In contrast to counting the number of minimal objects, counting the number of
minimum objects reveals some subtle facts about counting complexity classes. It
is indeed not an obvious task at all to check in polynomial time whether a
feedback arc set $F$ is of minimum size (intuitively this would require to
compute $m$ in polynomial time, but deciding if $m \leq k$ for a given $k$
is well known to be
$\NP$-complete \cite{k72}).
%However, in the definition of $\sPoly$ the
%certificate can be anything of polynomial size (not necessarily the feedback
%arc set only) so one has to be careful on how to formalize this intuition.
Fortunately the literature on counting problems provides appropriate tools
to characterize the complexity of {\bf \#Minimum-FAS}.

From standard techniques we derive the following.

\begin{theorem}
  \label{th:minimum-fas}
  {\bf \#Minimum-VC} $\redparsi$ {\bf \#Minimum-FAS},\\\
  {\bf \#Minimum-FAS} is $\sPoly$-hard and $\sPoly$-easy.
\end{theorem}

\begin{proof}
  The construction $G'(2)$ from the proof
  of Theorem~\ref{th:card-fas} is a polynomial parsimonious
  reduction from {\bf \#Minimal-VC} to {\bf \#Minimum-FAS},
  which is therefore $\sPoly$-hard by
  Theorem~\ref{th:minimum-vc-spoly}. Indeed, a minimum FAS in $G'(2)$ contains
  no arc corresponding to edges of $G$, hence there is a one-to-one
  correspondence between minimum FAS in $G'(2)$ and minimum VC in $G$.

  \medskip

  For the $\sPoly$-easiness, there is a straightforward
  polynomial Turing reduction to {\bf \#Card-FAS} which is in $\sPoly$
  (Theorem~\ref{th:card-fas}): one can simply call the oracle 
  with the digraph $G=(V,A)$ for $k=0,1,\dots$ until
  the first time its answer is not zero (which happens at least when $k=|V|$),
  this is the number of minimum FAS.
\end{proof}

Although the class $\sPoly$ is closed under polynomial parsimonious reductions
(meaning that $A \redparsi B$ and $B \in \sPoly$ implies $A \in \sPoly$), the
closure of $\sPoly$ for polynomial Turing reductions is an open 
question\footnote{Note that the analogous question regarding the
closure of $\NP$ under 
polynomial Turing reduction is also open, while $\NP$ is known to be closed
under polynomial many-one reductions.}
(and
is equivalent to $\Poly=\NP$, as we will see). For
more on closure properties of $\sPoly$ see
\cite{fr96,ottw96,oh93,tw92}
The authors of \cite{hp09} defined the complexity class $\sOptPoly$ (and
analogous classes for problems higher in the polynomial counting hierarchy) for
which minimality (or maximality) problems such as {\bf \#Minimum-FAS} are
provably\footnote{More precisely with a proof of {\em appropriate difficulty}, since
all these are complete if $\Poly=\NP$\dots} complete.

\begin{theorem}
  \label{th:minimum-fas-opt}
  {\bf \#Minimum-FAS} is $\sOptPoly$-complete.
\end{theorem}

\begin{proof}
  {\bf \#Minimum-FAS} is in $\sOptPoly$ since
  given $G=(V,A)$, one can guess nondeterministically a subset $F \subseteq A$,
  then check in polynomial time that the digraph $(V,A\setminus F)$ is acyclic.
  If it not acyclic then reject, if it is acyclic then output the size of $F$
  encoded in binary and accept.

  \medskip

  We already proved in Theorem~\ref{th:minimum-fas} that
  {\bf \#Minimum-VC} $\redparsi$ {\bf \#Minimum-FAS},
  hence from Theorem~\ref{th:minimum-vc-soptpoly} {\bf \#Minimum-FAS}
  is $\sOptPoly$-hard.
\end{proof}

Thanks to a result of \cite{hp09} we can nicely complement
Theorem~\ref{th:minimum-fas}.

\begin{corollary}
  \label{coro:minimum-fas}
  If {\bf \#Minimum-FAS} is in $\sPoly$ then $\Poly=\NP$.
\end{corollary}

\begin{proof}
  If {\bf \#Minimum-FAS} is in $\sPoly$ then $\sOptPoly=\sPoly$
  from Theorems~\ref{th:minimum-fas} and \ref{th:minimum-fas-opt}
  (it straightforwardly holds that $\sPoly \subseteq \sOptPoly$),
  hence the result follows from a direct application of
  \cite[Theorem~9]{hp09} for $k=0$.
\end{proof}

%%%%%%%%%%%%%%%%%%%%%%%%%%
\section{Complexity of counting feedback vertex sets}
\label{s:fvs}

There is a straightforward polynomial parsimonious reduction
from counting feedback arc sets
to counting feedback vertex sets in a digraph.
A {\em feedback vertex set} (FVS) of a digraph $G$ is a subset
of vertices intersecting every cycle of $G$ ({\em i.e.} containing,
for every cycle of $G$, at least one vertex from its arcs).
The reduction, as noted in \cite{ss02}, consists in considering the
{\em line digraph} of $G=(V,A)$, denoted $L(G)=(V',A')$ and defined as
$$
  V'=A \text{, and } ((u,v),(v,w)) \in A' \text{ for all }
  u,v,w \in V \text{ such that } (u,v),(v,w) \in A.
$$
Indeed, there is one-to-one correspondence between the cycles of $G$
and $L(G)$, and a one-to-one correspondence between the FAS of $G$
and the FVS of $L(G)$.
As a consequence ($\sPoly$ and $\sOptPoly$ being closed for $\redparsi$),
all the results of this note also apply to counting the number of
feedback vertex sets.

It is worth remembering that counting the number of minimum feedback vertex sets is
known to be $\sPoly$-hard even when restricted to
planar digraphs \cite{hmrs98}.

%%%%%%%%%%%%%%%%%%%%%%%%%%
\section{Acknowledgments}

The author is thankful to Madhav Marathe for helpful correspondence.
The author is affiliated to
Universit\'e C\^ote d'Azur, CNRS, I3S, UMR 7271, Sophia Antipolis, France,
and Aix Marseille Universit\'e, Universit\'e de Toulon, CNRS, LIS, UMR 7020, Marseille, France.
This work received financial support from
the \emph{Young Researcher} project ANR-18-CE40-0002-01 ``FANs'',
the project ECOS-CONICYT C16E01,
the project STIC AmSud CoDANet 19-STIC-03 (Campus France 43478PD).

\bibliographystyle{plain}
\bibliography{biblio}

\begin{thebibliography}{10}

\bibitem{fr96}
L.~Fortnow and N.~Reingold.
\newblock {PP} is closed under truth-table reductions.
\newblock {\em Information and Computation}, 124(1):1--6, 1996.

\bibitem{gj79}
M.~R. Garey and D.~S. Johnson.
\newblock {\em Computers and Intractability; A Guide to the Theory of
  NP-Completeness}.
\newblock W. H. Freeman, 1979.

\bibitem{hp09}
M.~Hermann and R.~Pichler.
\newblock Complexity of counting the optimal solutions.
\newblock {\em Theoretical Computer Science}, 410(38):3814--3825, 2009.

\bibitem{hmrs98}
H.~B. Hunt, M.~V. Marathe, V.~Radhakrishnan, and R.~E. Stearns.
\newblock The complexity of planar counting problems.
\newblock {\em SIAM Journal on Computing}, 27:1142--1167, 1998.

\bibitem{k72}
R.~M. Karp.
\newblock Reducibility among combinatorial problems.
\newblock In {\em Complexity of Computer Computations}, pages 85--103, 1972.

\bibitem{l86}
N.~Linial.
\newblock Hard enumeration problems in geometry and combinatorics.
\newblock {\em SIAM Journal on Algebraic Discrete Methods}, 7:331--335, 1986.

\bibitem{ottw96}
M.~Ogihara, T.~Thierauf, S.~Toda, and O.~Watanabe.
\newblock On closure properties of \#{P} in the context of {PF}$\circ$\#{P}.
\newblock {\em Journal of Computer and System Sciences}, 53(2):171--179, 1996.

\bibitem{oh93}
M.~Ogiwara and L.~A. Hemachandra.
\newblock A complexity theory for feasible closure properties.
\newblock {\em Journal of Computer and System Sciences}, 46(3):295--325, 1993.

\bibitem{pb83}
J.~S. Provan and M.~O. Ball.
\newblock The complexity of counting cuts and of computing the probability that
  a graph is connected.
\newblock {\em SIAM Journal on Computing}, 12:777--788, 1983.

\bibitem{ss02}
B.~Schwikowski and E.~Speckenmeyer.
\newblock On enumerating all minimal solutions of feedback problems.
\newblock {\em Discrete Applied Mathematics}, 117(1):253--265, 2002.

\bibitem{tw92}
S.~Toda and O.~Watanabe.
\newblock Polynomial-time 1-{T}uring reductions from \#{PH} to \#{P}.
\newblock {\em Theoretical Computer Science}, 100(1):205--221, 1992.

\bibitem{v79b}
L.~G. Valiant.
\newblock The complexity of computing the permanent.
\newblock {\em Theoretical Computer Science}, 8(2):189--201, 1979.

\bibitem{v79}
L.~G. Valiant.
\newblock The complexity of enumeration and reliability problems.
\newblock {\em SIAM Journal on Computing}, 8:410--421, 1979.

\end{thebibliography}

\end{document}